\documentclass[12pt]{article}
\usepackage[T1]{fontenc}
\usepackage{times} 
\usepackage[cp1250]{inputenc}  
\usepackage{amsmath}
\usepackage{amsfonts}
\usepackage{amsthm}
\def\Z{{\mathbb Z}}

\def\R{{\mathbb R}}
\def\N{{\mathbb N}}

\def\CA{{\mathcal A}}

\def\CH{{\mathcal H}}

\DeclareMathSymbol\crossrt{\mathrel}{AMSb}{"6E}
\DeclareMathSymbol\crosslt{\mathrel}{AMSb}{"6F}

\def\const{\hbox{const}}

\def\dm#1{(d\!-\!#1)}

\newtheorem{lemma}{Lemma}[section]

\newtheorem{proposition}[lemma]{Proposition}

\newtheorem{definition}[lemma]{Definition}

\newtheorem{remark}[lemma]{Remark}

\newbox\ncintdbox \newbox\ncinttbox
\setbox0=\hbox{$-$} \setbox2=\hbox{$\displaystyle\int$}
\setbox\ncintdbox=\hbox{\rlap{\hbox
    to \wd2{\kern-.1em\box2\relax\hfil}}\box0\kern.1em}
\setbox0=\hbox{$\vcenter{\hrule width 4pt}$}
\setbox2=\hbox{$\textstyle\int$}
\setbox\ncinttbox=\hbox{\rlap{\hbox
    to \wd2{\kern-.14em\box2\relax\hfil}}\box0\kern.1em}
\newcommand{\ncint}{\mathop{\mathchoice{\copy\ncintdbox}
    {\copy\ncinttbox}{\copy\ncinttbox}
    {\copy\ncinttbox}}\nolimits}

\title{Spectral action \\ for scalar perturbations of Dirac operators.}
\author{
Andrzej Sitarz\thanks{Partially supported by MNII grants 189/6.PRUE/2007/7 and N 201 1770 33}, 
Artur Zając\thanks{Supported by an International PhD Program ,,Geometry and Topology in Physical Models''} \\
{\em Institute of Physics, Jagiellonian University}\\
{\em Reymonta 4, 30-059 Krak\'ow, Poland} \\
}
\date{}
\begin{document}
\maketitle
\begin{abstract}
We investigate the leading terms of the spectral action for odd-dimensional 
Riemannian spin manifolds with the Dirac operator perturbed by a scalar 
function. We calculate first two Gilkey-de Witt coefficients and make explicit
calculations for the case of $n$-spheres with a completely symmetric Dirac. 
In the special case of dimension $3$, when such perturbation corresponds 
to the completely antisymmetric torsion we carry out the noncommutative 
calculation following Chamseddine and Connes and study the case of 
$SU_q(2)$.
\end{abstract}

\section{Introduction}
In the setup of noncommutative differential geometry developed from Alain Connes' 
idea of spectral triples \cite{CoRe} the fundamental ingredient of the 
construction is based on the spin geometry and the Dirac operator. Although 
more general Dirac-type operators are also admitted, the fundamental theorems 
\cite{CoRC,BVF} give equivalence only between commutative spectral triples 
and usual spin geometries. 

Usually, the Dirac operator on spin manifold is taken to be the operator, which 
comes from the Levi-Civita connection on the tangent bundle. It appears, however,
that a wide class of generalized operators, which come, for example, from connections 
with torsions also satisfies the axioms of spectral triples. Although in the classical
differential geometry this might be irrelevant as we can easily pass to the language
of differential geometry and select only such Diracs which come from the torsion-free 
connections, in the pure noncommutative situation this method is no longer available.

Therefore, the class of Dirac operators with torsion is much larger than the classical 
Diracs and their perturbation. Classically, if one assumes that torsion is totally 
antisymmetric (which has some natural geometric justification) the torsion Dirac operators 
appear only for dimension superior than $2$. Whether this is the case in the noncommutative 
situation is not clear. In three dimensions the totally antisymmetric torsion tensor has 
only one component, hence the perturbation of the Dirac operator is only by a function.

We show, that such perturbations are admissible in every odd dimension and calculate the
correction to the spectral action in the classical case as well as in the noncommutative
situation (the latter in dimension $3$, with the example of $SU_q(2)$. This extends recent
calculation of spectral action for compact manifolds with torsion (with and without 
boundaries) which were carried out in \cite{HPS,ILV}, though in dimesnions higher than 
$3$ the scalar perturbation has no clear geometrical meaning.

\subsection{Scalar perturbations of spectral triples}

Let us begin with the definition of a {\em real} spectral triple. 

\begin{definition}
Let be $\CA$ be an algebra, $\CH$ a Hilbert space and $\pi$ a faithful representation 
of $\CA$ on $\CH$. The geometric data of a real spectral triple is given by $(\CA,\pi,\CH,D,J)$ 
where $D$ is a selfadjoint unbounded operator on $\CH$, $J$ is antilinear unitary operator, 
an integer modulo $8$ (dimension of the spectral triple) and the following relations:
\begin{itemize}
\item $\forall a \in \CA$, $[D,\pi(a)]$ is a bounded operator,
\item $\forall a,b \in \CA$, $[J^{-1} \pi(a) J, \pi(b)] =0$, (CC)
\item $JD = \epsilon_D DJ$, $\epsilon_D = \pm1$ depending on the dimension of the triple, 
\item $\forall a,b \in \CA$, $\left[ \pi(a), [D, \pi(b)] \right] = 0$,
\item $J^2 = \epsilon_J 1$, $\epsilon_J= \pm 1$ depending on the dimension of the spectral triple,
\item if the dimension of the triple is even, then there exists $\gamma=\gamma^\dagger$,
      such that:
      $$ \gamma^2=1, \;\; D\gamma=-\gamma D, \;\; [\gamma,\pi(a)]=0, \;\; J \gamma = 
      \epsilon_\gamma \gamma J, $$
      where $\epsilon_\gamma=\pm 1$ depends also on the dimension of the spectral 
      triple.
\end{itemize}
\end{definition}

In addition to the listed above conditions the spectral triple must satisfy a series of
additional requirements: for their list and meaning, as well as the list of all signs
we refer to the literature \cite{BVF,CoRe}. In the case when $\CA$ is an algebra of 
smooth functions on a manifold we replace the commutator condition (CC) by demanding
that $\pi(a)^\dagger = J^{-1} \pi(a) J$ for every $a \in \CA$.

One of the earliest results, which motivated the construction of spectral geometry 
was the so-called reconstruction theorem (later, under certain conditions generalized 
to equivalence, see \cite{CoRe}), which stated that for a compact spin manifold $M$,
with the algebra $\CA=C^\infty(M)$, $\CH$ being the Hilbert space of square integrable 
sections of the spinor bundle, $J$ the implementation of the involution in the Clifford 
algebra, $\gamma$ the natural $\Z_2$ grading of the Clifford algebra and $D$ the Dirac 
operator associated to the Levi-Civita connection, the data $(\CA,\CH,D,J,\gamma)$ form 
a commutative spectral triple.

Now, we have:

\begin{proposition}
With the same data as in the definition above, in the case of odd $KO$ dimension, all
relations are satisfied if we replace $D$ by:
$$D_\Phi = D +\Phi + \epsilon_D J \Phi J^{-1},$$ 
where $\Phi=\pi(\phi)$, for a selfadjoint element $\phi=\phi^*$ of the algebra $\CA$.
\end{proposition}

 \begin{proof}
To verify our claim it suffices to verify the algebraic relations, which depend 
on the Dirac operator. For the remaining part of the axioms, spectral and analytic,
it suffices to observe that $\Phi$ is a bounded operator and hence $D_\Phi$ is a
bounded perturbation of $D$. First, we calculate $J D_\Phi$:

\begin{align*} 
J D_\Phi &= J (D +\Phi + \epsilon_D J \Phi J^{-1}) = 
\epsilon_D D J + J\Phi + \epsilon_J \epsilon_D \Phi J^{-1} \\
&= \epsilon_D (D + \epsilon_D J \Phi J^{-1} + \Phi) J.
\end{align*}

and then the order one condition:

\begin{align*} 
\left[ [D_\Phi, \pi(a)], JbJ^{-1} \right] &= 
\left[ [D +\Phi + \epsilon_D J \Phi J^{-1}, \pi(a)], JbJ^{-1} \right] \\
&= \left[ [\Phi, \pi(a)], JbJ^{-1} \right] = 0,
\end{align*}

where we have used that $\Phi \in \CA$ and conjugation by $J$ maps the 
elements of $\CA$ to its commutant.
\end{proof}

\begin{remark}
Note that in the case of manifolds, when $J$ maps elements of $\CA$ to $\CA$,
the perturbation $\Phi + \epsilon_D J \Phi J^{-1} = \Phi (1+\epsilon_D)$,
and it vanishes identically in the $KO$ dimension $1$ and $5$ modulo $8$.
\end{remark}

\section{The spectral action for manifolds with $D_\Phi$.}

We shall calculate here the leading three coefficients of the spectral
action for scalar perturbations of the Dirac operator on odd-dimensional
compact spin manifolds. We do not assume any conditions on the reality 
structure, hence the perturbation by a real function $\Phi$ is possible 
in any odd dimension.

Therefore $D_\Phi = D + \Phi$, $\Phi = C^\infty(M,\R)$. Let us recall the 
Schr\"odinger-Lichnerowicz formula, applied to $D_\Phi$:

$$
\begin{aligned}
(D_\Phi)^2 &=  D^2 + D\Phi + \Phi D + \Phi^2 \\
    &=  \Delta + \frac{1}{4} R + [D, \Phi] + 2 \Phi D + \Phi^2.
\end{aligned}
\label{TLichner}
$$
where $\Delta$ is the spinorial Laplacian. In order to calculate the 
leading term parts of the spectral action, we use the standard techniques 
to calculate the heat-kernel coefficients for the second-order differential
operators of the Laplace type over spin manifolds. We used both the explicit
formulas obtained by Barth \cite{Barth}, as well as the general results 
presented by Vassilevich \cite{Vas}. 

In our notation $d$ is the dimension of the manifold, $n$ denoted the dimension
of the fibres of the spinor bundle. The result, up to terms, which are total 
divergence (which vanish, and since we consider manifolds without boundary) is:

\begin{equation}
\begin{aligned}
\hbox{$[a_1]$} =& (4\pi)^{-\frac{d}{2}}\, n \, \left( - \frac{1}{12} R + (d-1) \phi^2 \right), \\
\hbox{$[a_2]$} =& (4\pi)^{-\frac{d}{2}}\, \frac{n}{180} 
\left( \frac{5}{2} R^2 - 4 R_{ij} R^{ij} - \frac{7}{2} R_{ijkl} R^{ijkl} \right. \\
& \phantom{\frac{5}{2}} + 120 (d-1)(d-3) \phi^4 + 60 (3-d) R \, \phi^2 \\
& \left. \phantom{\frac{5}{2}}  + 120 (d-1) (\nabla_i \phi)(\nabla^i \phi) \right).
\label{a_d}
\end{aligned}
\end{equation}
The above result has been obtained earlier (for $[a_1]$) by many authors \cite{KW,Gusynin,Donn} 
and also \cite{Yajima1,Yajima2} in the case of dimension $3$, where the scalar perturbation 
corresponds to torsion. The $[a_2]$ coefficient in dimension $3$ is again a special case
of Dirac operator with an antisymmetric torsion.

An interesting situation happens in dimension $3$.

\begin{lemma}
The heat kernel coefficients for scalar perturbation of Dirac in dimension $d=3$ read:
\begin{equation}
\begin{aligned}
\hbox{$[a_1]$} =& 2 (4\pi)^{-\frac{3}{2}}\,  \, \left( - \frac{1}{12} R + 2 \phi^2 \right), \\
\hbox{$[a_2]$} =& (4\pi)^{-\frac{3}{2}}\, \frac{8}{3} (\nabla_i \phi)(\nabla^i \phi).
\label{a_13}
\end{aligned}
\end{equation}
\end{lemma}

\begin{proof}
Indeed, observe that many terms vanish in the case of $d=3$. Additionally, 
both the Weyl tensor (and so its square) as well as the Gauss-Bonnet integrand 
must vanish:
\begin{equation}
\begin{aligned}
& R_{ijkl} R^{ijkl} - 2 R_{ij} R^{ij} + \frac{1}{3} R^2 = 0, \\
& R_{ijkl} R^{ijkl} - 4 R_{ij} R^{ij} + R^2 = 0,
\end{aligned}
\end{equation}
Using it we can see that the terms, which depend only on the Riemann and Ricci tensors and the
scalar curvature add up to zero. Hence only the kinetic term for $\Phi$ remains.
\end{proof}
 
Of course, in the case of three dimensional manifolds the coefficient $[a_2]$ will not be present 
in the leading terms of the perturbative expansion of the spectral action with respect to the cut-off 
parameter $\Lambda$. However, when considering, for example, the spectral action on $M \times S^1$ 
with the product geometry, it will appear as the scale invariant term. 

\subsection{Explicit spectral action for spheres}

In the special case of the $d$-dimensional spheres (still $d$ being an odd number), we
can independently calculate the perturbative expansion of the spectral action directly
from the spectrum of the Dirac. The method, based on the Poisson summation formula,
which was first presented in \cite{CoCh} and then used by \cite{MarPie} to study 
cosmological aspects of the spectral action will allow us to have explicit formulas for 
the case of $\phi = \const$.

Let us recall that the spectrum of the equivariant Dirac operator on a unit sphere
of an odd dimension $d$ is given ($\lambda_n$ is the eigenvalue, $N(n)$ denotes 
the multiplicity):

$$ \lambda_\pm(n)  = \pm (\frac{d}{2} +n), \;\;\; 
   N(n) = \frac{2^{\frac{d-1}{2}} (n+d-1)!}{n! (d-1)!}, \;\;n \geq 0. 
$$

When we consider the scalar perturbation of $D$ by a constant $t \in \R$, the
spectrum is, accordingly changed to:

$$ 
\lambda_\pm(n)  = \pm (\frac{d}{2} + n \pm t), \;\;\; 
N(n) = \frac{2^{\frac{d-1}{2}} (n+d-1)!}{n! (d-1)!}, \;\;n \geq 0.
$$

Let us write an explicit formula we have for the spectral action defined by 
the cut-off function $f$: 

$$ {\cal S}_{S^d}(t) = 
\sum_{n \geq 0} N(n) \left( 
  f \left( \frac{ n+\frac{d}{2}+t}{\Lambda} \right) 
+ f \left( \frac{-n-\frac{d}{2}+t}{\Lambda} \right) 
\right). 
$$

Observe that:
$$ -n-\frac{d}{2}+t  = (-n-d)+ \frac{d}{2} + t, $$
and that 
$$ N(n) = N(-n-d), $$
as $d$ is odd, so $d-1$ is even. Therefore, we we can rewrite it as 
$$ {\cal S}_{S^d}(t) =
  \sum_{n \in \Z} N(n) f \left( \frac{n+\frac{d}{2}+t}{\Lambda} \right).
$$  
Observe that the terms for $n=-1,\ldots,-d+1$ do not appear as these
are zeros of the function $N$ counting multiplicities.

Now, we use the Poisson summation formula to the sum over all integers. 
Define $g(x) = f(x+\alpha)$, $\alpha \in \R$. 
We have:
$$ \hat{g}(k) = \frac{1}{2\pi} \int g(x) e^{-ikx} = 
                \frac{1}{2\pi} \int f(x+\alpha) e^{-ikx}
              = \hat{f}(k) e^{ik\alpha}.   
$$ 
and
$$ \hat{f}^{(p)}(k) = \frac{1}{2\pi} \int (-ix)^p f(x) e^{-ikx} = 
   \frac{1}{2\pi} \int (-i(x+\alpha))^p f(x+\alpha) e^{-ikx} e^{-ik\alpha},   
$$ 
where $\hat{f}^{(p)}(k)$ denotes $(\partial_k)^p \hat{f}(k)$.

Using now the expansion in powers of $\Lambda$ up to terms of order 
$o(\Lambda^{-1})$:
$$ \sum_{n \in \Z} g\left(\frac{n}{\Lambda}\right) = \Lambda \hat{g}(0) + o(\Lambda^{-1}), $$ 
and for the derivatives:
$$ \sum_{n \in \Z} (-in)^k g\left(\frac{n}{\Lambda}\right) = \Lambda^{k+1} \hat{g}^{(k)}(0) + o(\Lambda^{-1}), $$ 

We can now use the following lemma, in order to calculate the three leading terms 
of the perturbative spectral action.

\begin{lemma}
For each $t \in \R$ we have:
$$ 
\begin{aligned}
&(n+d-1)(n+d-2)\cdots(n+1) = (n+\frac{d}{2} +t)^{d-1} 
  - \dm{1} t (n+\frac{d}{2} + t)^{d-2} \\
& + \left( \frac{1}{2}\dm{1}\dm{2} t^2 - \frac{1}{24}d\dm{1}\dm{2} \right)(n+\frac{d}{2} +t)^{d-3} \\
& + \left( \frac{1}{24}d\dm{1}\dm{2}\dm{3}t - \frac{1}{6}\dm{1}\dm{2}\dm{3} t^3 \right) (n+\frac{d}{2} +t)^{d-4} \\
& + \left(\frac{1}{24}\dm{1}\dm{2}\dm{3}\dm{4} t^4  
- \frac{1}{48}d \dm{1}\dm{2}\dm{3}\dm{4} t^2 \right. \\
& \phantom{xx} \left.+ \frac{1}{5760}d \dm{1}\dm{2}\dm{3}\dm{4}(5d\!+\!2) \right) 
 (n+\frac{d}{2} +t)^{d-5} + o(n^{d-5}).
\end{aligned}
$$
\end{lemma}

We skip the proof based on explicit calculations. 

For an even function $f$ the terms of the order $\Lambda^{d+1-2k}$ shall 
vanish (as the value of the Fourier transform at $0$ is just the integral 
of an odd function over $\R$. Therefore the first three nontrivial terms
are (up to the overall sign $(-1)^{\frac{d-1}{2}}$, which arises from 
the Fourier transform):

\begin{equation} 
\begin{aligned}
&{\cal S}_{S^d}(t) = \Lambda^d  \hat{f}^{(d-1)}(0) \\
 &\phantom{x}-\Lambda^{d-2} \frac{\dm{1}\dm{2}}{2} \left( t^2 - \frac{d}{12} \right) \hat{f}^{(d-3)}(0)\\
 &\phantom{x}+\Lambda^{d-4}\frac{\dm{1}\dm{2}\dm{3}\dm{4}}{24}  
 \left( t^4 - \frac{d}{2} t^2 + \frac{d (5d\!+\!2)}{240} \right) \hat{f}^{(d-5)}(0).
 \end{aligned}
 \label{poiss}
 \end{equation}
 
We can easily compare this result with the one obtained previously, which could be treated
as the independent test that the results are correct. Let us recall that for the unit 
$d$-dimensional sphere the following identities are true:

$$ (R_{ijkl})^2 = 2d\dm{1}, \;\;(R_{ij})^2=d\dm{1}^2, \;\;\;R=d\dm{1}. $$

Therefore, using the formula (\ref{a_d}) we calculate:

\begin{equation}
\begin{aligned}
\hbox{$[a_1]$} & \sim (-\frac{1}{12}d\dm{1} + \dm{1}t^2 ) = \dm{1} (t^2 - \frac{d}{12}), \\
\hbox{$[a_2]$} & \sim (120 \dm{1}\dm{3}t^4 - 60d\dm{1}\dm{3}t^2 \\
& \phantom{xxxxxxxx} + \frac{5}{2}d\dm{1} -7d\dm{1}
-4 d \dm{1}^2) \\
& \phantom{xxx} = 120 \dm{1}\dm{3} \left(t^4  - \frac{d}{2} t^2 +  \frac{d(5d\!+\!2)}{240} \right).
\end{aligned}
\end{equation}

As we can see it is exactly the result obtained above in (\ref{poiss}).

\subsection{Noncommutative spectral action}

Assume that we have a real spectral triple $(\CA,\CH,D)$ with a simple 
dimension spectrum. By analogy with the commutative example we propose 
to call a perturbation of the Dirac 
operator by a selfadjoint element of the algebra, $\CA \ni \Phi = \Phi^*$, 
a torsion-perturbed Dirac, 
$D_\Phi = D + \Phi + J \Phi J^{-1}$ (recall that $JD=DJ$ in dimension $3$). 
It is clear that $D_T$ is still a good Dirac operator for the real spectral 
triple\footnote{The only possible exception is the axiom of the existence of the 
Hochschild cycle. However, this fails also in the case of the {\em usual} 
perturbation of the Dirac by a one-form.}. We call $F=\hbox{sign}(D)$ and 
assume that $[F,a] \in OP^{-\infty}$  for any $a \in \CA$. For simplicity we
assume that the kernel of $D_\Phi$ is empty (if it is not the case one can
correct $D_\phi$ by a finite rank operator, which is a projection on the kernel,
for details see \cite{EILS}).

We have:
\begin{proposition}
\label{TorSpAc}
The coefficients of the full perturbative spectral action on a real spectral 
triple $(\CA,\CH,D_\Phi)$ are:
\begin{align*}
&(a) \hspace{5mm} \ncint \vert D_{\Phi}\vert^{-3}  = \ncint |D|^{-3}.\\
&(b) \hspace{5mm} \ncint \vert D_{\Phi}\vert^{-2}  = \ncint
|D|^{-2} - 4 \ncint \Phi F |D|^{-3}. \\
&(c) \hspace{5mm}\ncint \vert D_{\Phi}\vert^{-1}  
=    \ncint |D|^{-1} -2 \ncint \Phi F |D|^{-2} + 2 \ncint \Phi^2 |D|^{-3}
  + 2 \ncint \Phi J \Phi J^{-1} |D|^{-3}.\\
&(d)  \hspace{5mm} \zeta_{D_{\Phi}}(0)  - \zeta_{D}(0) =
   2 \ncint \Phi D^{-1} 
   - \ncint \Phi (\Phi+J\Phi J^{-1}) D^{-2} \\ 
& \hspace{1cm}  
   - \ncint  [D,\Phi] (\Phi + J \Phi J^{-1}) D^{-3}
   + \tfrac{2}{3} \ncint \Phi^3 D^{-3} 
   +  2 \ncint \Phi^2 J \Phi J^{-1}  D^{-3}.
\end{align*}
\label{prop23}
\end{proposition} 

\begin{proof}
We use the results obtained in \cite{EILS}. First, from Proposition 4.9 we directly obtain (a), then using the proof of Lemma 4.10 (applied to our perturbation of the Dirac) gives us (b)-(c). Similarly, Lemma 4.5 \cite{EILS},
when applied to any perturbation of the Dirac operator of order $0$, say $C$,
gives:
$$
\ncint \zeta_{D_{C}}(0) - \zeta_{D}(0)
= S(C) = \ncint  C D^{-1} - \frac{1}{2} \ncint (C D^{-1})^2 +
\frac{1}{3} \ncint   (C D^{-1})^3.
$$
Taking $C=\Phi + J \Phi J^{-1}$ and using that the noncommutative integral
is invariant under conjugation by$J$ we obtain  (d).
\end{proof}

It is interesting to see the application of the above calculations in the
commutative case of a three-dimensional manifold $M$. 
Then $\Phi=\Phi^*=J \Phi J^{-1}$ is a function on $M$. To see that 
the spectral action simplifies significantly in this case we first observe: 

\begin{lemma}\label{intFl}
Let $D$ be a Dirac operator on $M$ and $F= \hbox{sign}(D)$. Then
$$ \ncint \phi F |D|^{-3} = 0, $$
for any function $\phi \in C^\infty(M)$.
\end{lemma}

\begin{proof}
Indeed, since we are dealing with pseudodifferential operators (both $f$ 
and $F$ are of order $0$, $D$ is of order $1$) we can use the symbols 
and the relation of the noncommutative integral with the Wodzicki 
residue. In fact, when we rewrite $f F |D|^{-3}$ as $f D D^{-4}$
we see that the principal symbol of this expression is scalar multiple 
of the symbol of $D$ (as the leading symbol of $D^{-4}$ is scalar). Hence its 
Clifford trace vanishes.
\end{proof}

Finally, we can show:
\begin{proposition}
The leading terms of the perturbative expansion of the spectral action 
for the Dirac operator with torsion on the three-dimensional manifold 
are: 
\begin{equation}
\begin{aligned}
& \Lambda^3 \hbox{\ term:\ } \,\, \sim \hbox{Vol}(M) \\
& \Lambda^1 \hbox{\ term:\ } \,\, \sim \frac{1}{2\pi^2} \left( -\frac{1}{12} \int_M R 
                                      + 8 \int_M \Phi^2  \right).
\end{aligned}
\end{equation}
in particular, there are no $\Lambda^2$ terms and scale-invariant contributions. 
\end{proposition}              
 
\begin{proof}
First, observe that using Lemma \ref{intFl} (and extending it slightly) we can 
at once say that these terms must vanish:
$$ \ncint \Phi F |D|^{-3}, \, \, \ncint \Phi^3 F |D^{-1}|. $$ 
Furthermore, repeating the arguments of the symbol calculus for the Dirac 
operator we see that, on any manifold of dimension $3$:
$$ \ncint \Phi D^{-1}, \,\, \ncint \Phi^3 D^{-3},$$
must also vanish.

A bit more work is necessary to show that
$$\ncint  \Phi [D,\Phi] D^{-3}, \,\, \ncint  \Phi^2 D^{-2}, $$
vanish. Here, one can use explicit calculations by Kastler \cite{Kast}. Using the
explicit form of the paramatrix for $D^{-2}$ we see that that its component of 
order $-3$, $\sigma_{-3}(\xi,x)$ is an odd polynomial in $\xi$. Therefore the 
last integral in the above list is $0$. For the first one we need to use the trace 
property of the noncommutative integral. 

Finally, using the identities:

$$ 
\begin{aligned}
& \ncint |D^{-3}| = \frac{1}{\pi^2} \hbox{Vol}(M), \\
& \ncint \psi |D^{-3}| = \frac{1}{\pi^2} \int_M \psi, \\
& \ncint |D^{-1}| = -\frac{1}{24\pi^2} \int_M R(g) 
\end{aligned}
$$
where $\psi \in C^\infty(M)$ and $R(g)$ is the scalar curvature of $M$,
we obtain the result (\ref{a_13}) (recall that $\phi = \frac{1}{2} \Phi$). 
\end{proof}
 
\section{The quantum sphere $SU_q(2)$}
 
We shall show here that the spectral action terms for the quantum deformation
of the three-sphere are not much different from the undeformed case.  We 
start with the undeformed invariant Dirac operator and the spectral triples 
as described in \cite{Naiad} (we omit here details of the notation and 
construction of spectral geometry).

We take a torsion term $\Phi$ being a finite sum of homogeneous polynomials 
in $a,a^*,b$:
\begin{equation}
\Phi = \sum C_{\alpha, \beta, \gamma} a^\alpha b^\beta (b^*)^\gamma, 
\label{fphi}
\end{equation}
where $\alpha \in \Z$ (if $\alpha<0$ then we take $(a^*)^{|\alpha|}$) 
and $\beta,\gamma \in \N$. For simplicity we do not take into account 
the condition of $J$-reality, thus we restrict ourselves only to $D + \Phi$.
If $\Phi$ is selfadjoint then:
$$ C_{\alpha, \beta, \gamma} = q^{\alpha(\gamma+\beta)} \overline{C_{-\alpha, \gamma, \beta}}.$$ 

\begin{proposition}
\label{suqSpAc}
The coefficients of the leading terms of the perturbative spectral action on 
a real spectral triple with arbitrary torsion $\Phi$ over 
$SU_q(2)$ are: 

\begin{align*}
&(a) \hspace{5mm} \ncint \vert D_{\Phi}\vert^{-3}  = 2.\\
&(b) \hspace{5mm} \ncint \vert D_{\Phi}\vert^{-2}  = 0.\\
&(c) \hspace{5mm}\ncint \vert D_{\Phi}\vert^{-1}  
=  -\frac{1}{2} + \ncint \Phi^2 |D|^{-3}.
\end{align*}  
where we have used that for the standard Dirac operator on $S^3$:
$$ 
   \ncint \vert D\vert^{-3} = 2, \;\;\; 
   \ncint \vert D\vert^{-1} = -\frac{1}{2}.
$$
For the form of $\Phi$ as assumed in (\ref{fphi}) we have:
$$ \ncint \Phi^2 |D|^{-3} = \sum _\alpha |C_{\alpha, 0, 0}|^2. $$
\end{proposition}

\begin{proof}
Using the results for the classical Dirac operator and the noncommutative expansion
from Proposition \ref{prop23} we obtain (a) and the part of (b). The vanishing of 
the terms linear in $\Phi$ in (b) and (c) is a consequence of Theorem 3.4 \cite{ILS}. 

To calculate the noncommutative integral of $\Phi^2 |D|^{-3}$ we use the explicit 
form of it from \cite{ILS}.
\end{proof}

\begin{remark}
Observe that unlike in the classical case the second term of the spectral action 
has no single minimum as the moduli space of scalar perturbations, for which it 
reaches the minimum value is infinite dimensional (as it contains all functions, 
which have nontrivial dependence on $b$ or $b^*$). Therefore, the proposed spectral
condition for the vanishing of torsion makes no sense in the considered noncommutative 
geometry of $SU_q(2)$.
\end{remark}

Let us turn our attention to the terms, which vanish identically in 
the classical situation.
\begin{proposition}
\label{suqSpAcIn}
The scale invariant part of the action does not vanish and we have:
\begin{align*}
&  \hspace{5mm} \zeta_{D_{\Phi}}(0)  - \zeta_{D}(0) = \\
& \hspace{2cm}   =   - \frac{1}{2} 
\sum_{\alpha \in \Z}
\sum_{l \geq 0}
\sum_{0 \leq m,n \leq l}
\overline{C_{\alpha, l-n,m}} C_{\alpha, l-m, n} \\
& \hspace{15mm}
\left( \sum_{k=0}^{\alpha} (-1)^k q^{k(k-1)} 
\left[ \alpha \atop k \right]_{q^2} 
\frac{4}{1-q^{2(k+l)}} \right). 
\end{align*}
where 
$$ 
\left[ n \atop k \right]_{q^2} = { \prod_{i=1}^n (1-q^{2i}) \over
\left( \prod_{i=1}^k (1-q^{2i})  \right)  \left( \prod_{i=1}^{n-k} (1-q^{2i})  \right)}, 
$$
is the quantum binomial.
\end{proposition} 

\begin{proof}
First, let us see that using the same expansion as from the proof of proposition
\ref{prop23} the scale invariant terms are:
$$ 
\ncint \Phi D^{-1} 
  - \frac{1}{2} \ncint \Phi^2 D^{-2}  
  + \frac{1}{2} \ncint \Phi [D,\Phi]  D^{-3} 
  + \frac{1}{3} \ncint \Phi^3 D^{-3}. 
$$

Let us begin with the first term, $\ncint \Phi D^{-1}$. It is clear 
that only the terms with $(bb^*)^n$ can possibly contribute. The 
shortest way to show that the integral vanishes:
$$ \ncint (bb^*)^n D^{-1} = 0, $$
is to use the property of the following algebra automorphism 
of $SU_q(2)$:
\begin{align}
\rho(a):=a, \;\; \rho(a^*):=a^*,\;\; \rho(b):=b^*,\;\; \rho(b^*):=b.
\label{autominv}
\end{align} 

Using the Lemma 5.17 of \cite{ILS}, we see that for any 
homogeneous polynomial in $a,a^*,b,b^*$ the noncommutative 
integral:
$$
\ncint p(a,a^*,b,b^*) \, D^{-k} = 
 (-1)^{k} \ncint p(\rho(a),\rho(a^*),\rho(b),\rho(b^*) D^{-k},
$$
is (up to sign) invariant with respect to $\rho$. Since in our case
$k=1$ and $\rho(bb^*)=bb^*$ the integral must vanish. In fact, this 
proves also that the last term, $\frac{1}{3} \ncint \Phi^3 D^{-3}$, 
vanishes as well.

For the second term of the scale invariant part we again use 
Lemma 5.14 of \cite{ILS}. First of all, we calculate $\Phi^2$:
$$ 
\Phi^2 = \sum_{\alpha,\beta,\gamma,\alpha',\beta',\gamma'}  
C_{\alpha, \beta, \gamma} C_{\alpha', \beta', \gamma'}
a^{\alpha} b^\beta (b^*)^\gamma
a^{\alpha'} b^{\beta'} (b^*)^{\gamma'}. 
$$

Then, the diagonal part is the sum of elements where 
$\alpha'=-\alpha$ and $\beta+\beta' = \gamma+\gamma'$:

$$ 
\begin{aligned}
\Phi^2 &\sim_{\hbox{diag}} 
\sum_{\alpha \in \Z}
\sum_{l \geq 0}
\sum_{0 \leq m,n \leq l}
C_{-\alpha, m, l-n} C_{\alpha, l-m, n}
a^{-\alpha} b^{m} (b^*)^{l-n}
a^{\alpha} b^{l-m} (b^*)^{n} \\
&=
\sum_{\alpha \in \Z}
\sum_{l \geq 0}
\sum_{0 \leq m,n \leq l}
\overline{C_{\alpha, l-n,m}} C_{\alpha, l-m, n}
a^{\alpha} a^{-\alpha} (bb^*)^{l}.
\end{aligned}
$$

Next, the commutation relations for $SU_q(2)$ yields:
$$ 
\begin{aligned}
a^{-\alpha} a^\alpha =& (1-q^{2\alpha-2} bb^*)(1-q^{2\alpha-4} bb^*) \cdots (1-bb^*) \\
=& \sum_{k=0}^\alpha (-1)^k q^{k(k-1)} \left[ \alpha \atop k \right]_{q^2} (bb^*)^k. 
\end{aligned}
$$
Next, using it and Lemma 5.14 we obtain:
$$ 
\begin{aligned}
\ncint \Phi^2 D^{-2} &= 
\sum_{\alpha \in \Z}
\sum_{l \geq 0}
\sum_{0 \leq m,n \leq l}
\overline{C_{\alpha, l-n,m}} C_{\alpha, l-m, n} \\
&\hspace{15mm} \left( \sum_{k=0}^{\alpha} (-1)^k q^{k(k-1)} 
\left[ \alpha \atop k \right]_{q^2} 
\frac{4}{1-q^{2(k+l)}} \right). 
\end{aligned}
$$   
The next component of the scalar invariant part vanishes:
$$ \ncint [D,\Phi] \Phi D^{-3} = 0. $$
This is an easy consequence of the fact that one can rewrite 
the above expression as
$$ \ncint \delta(\Phi) \Phi |D|^{-3} = \frac{1}{2} \ncint \delta(\Phi^2) |D|^{-3}, $$
and then, using expansion of diagonal terms of $\Phi^2$ we obtain the desired result.
\end{proof}
\section{Conclusions}
We have seen that in the case of three-dimensional manifolds the torsion term 
contributes only to the spectral action through a quadratic term. Therefore,
the postulate to restrict to torsion-free geometries by minimizing the term,
which comes from the noncommutative integral of $|D|^{-n+2}$, at a fixed 
metric appeared plausible. Indeed, it works in the classical commutative 
case and might be extended to other torsion-type perturbation of the Dirac
operators. 

The scalar type perturbations considered in this paper only in dimension $3$ 
have the interpretation of torsion and appear to have no direct geometrical
meaning in higher dimensions. We have shown that in the case of lowest possible
dimension for which they are nontrivial the next term spectral action includes 
coupling of the scalar field to the scalar curvature and the standard kinetic 
term, thus making the field dynamical. In higher dimensions an extra selfinteraction 
quartic term appears thus making it possible that in the first three terms one 
might have a nonzero vacuum value of the scalar perturbation.

The analysis of the spectral action for the genuine noncommutative example yields
even more interesting result. Clearly, for the $SU_q(2)$ it is no longer possible 
to eliminate the torsion by minimizing the $|D|^{-1}$ term of the spectral action.
In addition, unlike in the classical case the scale invariant terms do not vanish.
It appears that there is no reason to eliminate the scalar perturbation from the
considerations in the setup of noncommutative geometry. Its geometrical meaning 
and consequences for the model building in physics are still to be discussed.

\end{document}